\documentclass[aps,pra,reprint,groupedaddress]{revtex4-1}

\usepackage{graphicx} %
\usepackage{amsmath}
\usepackage{amssymb}
\usepackage{amsfonts}
\usepackage{hyperref}
\usepackage{xcolor}

\newtheorem{theorem}{Theorem}

\newtheorem{definition}{Definition}
\newtheorem{proof}{Proof}

\begin{document}

\title{Coherently tracking the covariance matrix of an open quantum system}

\author{Zibo Miao$^1$}
\email[]{zibo.miao@anu.edu.au}
\author{Michael R. Hush$^2$}
\author{Matthew James$^1$}
\affiliation{$^1$ ARC Centre for Quantum Computation and Communication Technology, Research School of Engineering, The Australian National University}

\affiliation{$^2$ School of Information Technology and Electrical Engineering, University of New South Wales at the Australian Defence Force Academy}

\begin{abstract}
Coherent feedback control of quantum systems has demonstrable advantages over measurement-based control, but so far there has been little work done on \emph{coherent estimators} and more specifically \emph{coherent observers}. Coherent observers are input the coherent output of a specified quantum plant, and are designed such that some subset of the observer and plant's expectation values converge in the asymptotic limit. We previously developed a class of mean tracking (MT) observers for open harmonic oscillators that only converged in mean position and momentum; Here we develop a class of covariance matrix tracking (CMT) coherent observers that track both the mean and covariance matrix of a quantum plant. We derive necessary and sufficient conditions for the existence of a CMT observer, and find there are more restrictions on a CMT observer than there are on a MT observer. We give examples where we demonstrate how to design a CMT observer and show it can be used to track properties like the entanglement of a plant. As the CMT observer provides more quantum information than a MT observer, we expect it will have greater application in future coherent feedback schemes mediated by coherent observers. Investigation of coherent quantum estimators and observers is important in the ongoing discussion of quantum measurement; As they provide estimation of a system's quantum state without explicit use of the measurement postulate in their derivation.

\end{abstract}


%02.30.Yy	Control theory
%05.45.Xt	Synchronization; coupled oscillators
%03.65.Ta	Foundations of quantum mechanics; measurement theory (for optical tests of quantum theory, see 42.50.Xa)
%03.65.Yz	Decoherence; open systems; quantum statistical methods (see also 03.67.Pp in quantum information; for decoherence in Bose-Einstein condensates, see 03.75.Gg)
%89.20.Kk	Engineering (for electrochemical engineering, see 82.47.Wx; for biomedical engineering, see 87.85.-d; for reservoir engineering in geothermal energy, see 88.10.G-; for nuclear engineering, see 28.00.00)


\maketitle

\section{Introduction}
Quantum engineering has seen rapid growth in the last two decades. Physicists, mathematicians and engineers have been working in unison to control a number of diverse systems in the quantum regime \cite{Sayrin:2011,Ruskov:2005, Hopkins:2003, Hush:2013, Bushev:2006, Thomsen:2002, Stockton:2004,HRHJ13}. Quantum control involving feedback has become particularly topical \cite{WM10,DP10,ZJ12S,PZMGUJ14}, as using information gained from a system can lead to more stable operation of a control protocol \cite{Szigeti:2013,James:2004}. Quantum feedback can be split into two paradigms: measurement-based and coherent feedback. Measurement-based feedback involves some measurement step in the feedback loop \cite{WM10,JNP08}, unfortunately measurements of quantum systems are typically slow and noisy as they involve coupling small quantum systems to macroscopic read out devices. Coherent feedback on the other hand is feedback where the controller and system are coupled directly without a measurement step \cite{SL00,NJP09,NY14}. The advantage is that the time scales of the controller and system can be made very similar as they are on the same scale. But beyond this practical advantage, there is increasing evidence that retaining the coherence of the feedback signal provides an intrinsic advantage over measurement-based feedback \cite{HM08,JNP08,WM94b,SL00}.

Coherent feedback is still a relatively new paradigm, and as such it lacks many of the tools commonly used in classical, or for that matter, other quantum feedback schemes. In particular, there are still only a limited number of options for coherently estimating a state within a feedback loop. It is well established classically that estimation using Kalman filters can provide improved performance over direct feedback schemes \cite{Stengel:1994}, and similar demonstrations have been performed for measurement-based quantum feedback \cite{Doherty:1999}. Unfortunately, traditional techniques do not appear to be applicable to coherent feedback due to difficulties with quantum conditioning onto non-commutative subspaces of signals. Instead, the closest option is so-called coherent quantum observers \cite{MJ12}. We previously developed a class of coherent quantum observers (as shown in Fig. \ref{Plant_CMTObs}), which can estimate the observables of linear and bilinear quantum plants described by quantum stochastic differential equations (QSDEs) in the sense of mean values, independent of any additional quantum noise in the observer \cite{MJ12,MEPUJ13}. We have proved that MT coherent observers can always be found, consistent with the laws of quantum mechanics, if the plant is detectable. In some cases the estimation of mean-values is sufficient and feedback can be improved with a coherent quantum observer. However, in many cases, the energy, correlations and indeed entanglement of the observed system may be the target of control or needed for feedback, for which the coherent quantum observer would not provide a reliable estimate. To remedy this issue, we propose to develop a modified coherent observer to track mean values, variances and correlations, namely, the CMT coherent observer.

In general, a CMT coherent observer outperforms a MT coherent observer in several respects. For instance, a CMT coherent observer allows us to achieve the most similar quantum state to that of the plant. Furthermore, it is well established that for a two-mode linear Gaussian system, the quantum correlations is completely characterized by the first and second moments \cite{AI05}, and thus entanglement can be mimicked by the utility of a CMT coherent quantum observer in this situation. Therefore, one can conclude a CMT observer can provide a better estimate in most cases. Nonetheless, we find that the error convergence rate of a CMT coherent observer can not be made arbitrarily high, plus we cannot guarantee that CMT coherent observers exist for systems where mean value coherent observers exist.  
\begin{figure}[!htp]
\centering
\includegraphics[scale=.45]{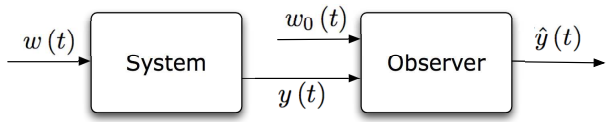}
\caption{A quantum plant and the corresponding coherent quantum observer in a cascade arrangement}
\label{Plant_CMTObs}
\end{figure}

The paper is organized as follows. We begin in Section \ref{sec:OHOALQ} by presenting the linear quantum state space model for open harmonic oscillators in the Heisenberg picture. In Section \ref{sec:QPACQO}, we briefly discuss quantum plants and (MT) coherent quantum observers. In Section \ref{sec:CCOFOHO}, we analyze the existence of CMT coherent observers, and show theorems which tell us how to construct CMT observers to be consistent with the laws of quantum mechanics. This is followed by numerical simulations in Section \ref{QCANS}, which illustrate the design and performance of CMT observers. Section \ref{sec:D} provides some concluding remarks and future research directions. The mathematical notation we use is defined in Appendix \ref{sec:app}. 

\section{Open harmonic oscillators and linear QSDES} 
\label{sec:OHOALQ}

The dynamics of an open quantum system are uniquely determined by the parametrization $(S,L,H)$ \cite{KP92,GJ09J1,JG10}. The self-adjoint operator $H$ is the Hamiltonian describing the self-energy of the system. The unitary matrix $S$ is a scattering matrix, and the column vector $L$  with operator entries is a coupling vector. $S$ and $L$ together specify the interface between the system and the fields. In the physics literature, it is common practice to describe open quantum systems  using a master equation for a density operator $\rho$, and it can easily be obtained from the triple $(S,L,H)$; indeed, we have
\begin{align}
d\rho=\left(i\left[\rho,H\right]+\mathcal{L^{\ast}}\left(\rho\right)\right)dt
\label{eq:mastereq}
\end{align}
where $\mathcal{L}^{*}\left(\rho\right)=L^{T}\rho L^{\sharp}-\frac{1}{2}L^{\sharp}L^{T}\rho-\frac{1}{2}\rho L^{\sharp}L^{T}$ (notation defined in Appendix \ref{sec:app}) and we assume natural units are being used. Given an operator $X$ defined on the initial Hilbert space $\mathsf{H}$, its Heisenberg evolution is defined by
\begin{align}
dX=&\left(\mathcal{L}\left(X\right)-i\left[X,H\right]\right)dt+dW^{\dagger}S^{\dagger}\left[X,L\right]\nonumber\\
&+\left[L^{\dagger},X\right]SdW+\mathrm{tr}\left[\left(S^{\dagger}XS-X\right)d\Lambda_{w}\right].
\label{eq:SLH}
\end{align}
With
\begin{align}
\mathcal{L}\left(X\right)=\frac{1}{2}L^{\dagger}\left[X,L\right]+\frac{1}{2}\left[L^{\dagger},X\right]L,
\label{eq:Lindblad}
\end{align}
which is called the Lindblad superoperator (Note $\mathcal{L}^{*}\left(\cdot\right)$ is the adjoint superoperator of $\mathcal{L}\left(\cdot\right)$). The operators $W$ are defined on a particular Hilbert space called a Fock space $\mathsf{F}$. When the fields (the number of fields is $n_w$) are in the vacuum states, these are the quantum Wiener processes which satisfy the It\^o rule
\begin{align*}
dWdW^{\dagger} = I_{n_w}dt.
\end{align*}
Input field quadratures $W+W^{\sharp}$ and $-i\left(W-W^{\sharp}\right)$ are each equivalent to classical Wiener processes, but do not commute. A field quadrature can be measured using homodyne detection \cite{GZ00,GJ09J1}. The gauge processes $\Lambda_{w}$ are input signals to the system as well.

We assume there is no interaction between different fields, and thus hereafter we assume $S$ to be the identity matrix without loss of generality \cite{JNP08}. This assumption eliminates the first time on the right hand side of Eq.~\eqref{eq:SLH}. To be specific,
\begin{align}
dX = & \left(\mathcal{L}\left(X\right)-i\left[X,H\right]\right)dt \nonumber\\
&+ \frac{1}{2}\left([X,L]-[X,L^\dagger]
\right)dW_1  \nonumber \\
&-\frac{i}{2}\left([X,L]+[X,L^\dagger]
\right)dW_2,
\label{eq:SLHqua}
\end{align}
with 
\begin{align*}
\left[\begin{array}{c} W_1 \\  W_2 \end{array}\right] = \left[\begin{array}{c} W+W^\sharp  \\ -i(W-W^\sharp) \end{array}\right].
\end{align*}

The quadrature form of the output fields is given by
\begin{align}
\left[\begin{array}{c} dY_1 \\  dY_2 \end{array}\right] = \left[\begin{array}{c} L+L^\sharp  \\ -i(L-L^\sharp) \end{array}\right] \,dt
+ \left[\begin{array}{c} dW_1 \\ dW_2 \end{array}\right].
\label{eq:outfields}
\end{align}

In this work we focus on \emph{open harmonic oscillators}. The dynamics of each oscillator is described by two Hermitian operators position $q_i$ and momentum $p_i$, which satisfy the canonical commutation relations $[q_i,p_j] = 2i \delta_{ij}$ where $\delta_{ij}$ is the Kronecker delta. For our purposes, it is convenient to collect the position and momentum operators of the oscillators into an $n$-dimensional column vector $x\left(t\right)$, defined by  $x\left(t\right)=\left(q_{1}\left(t\right),p_{1}\left(t\right),q_{2}\left(t\right),p_{2}\left(t\right), \ldots,q_{n}\left(t\right),p_{n}\left(t\right)\right)^{T}$. In this case the commutation relations can be re-written as:
 \begin{equation}
 x\left(t\right)x\left(t\right)^{T}-\left(x\left(t\right)x\left(t\right)^{T}\right)^{T}=2i\Theta_n \label{cancom}
 \end{equation}
 where $\Theta_{n} = I_{\frac{n}{2}}\otimes J$ with $J=\left[\begin{array}{cc}0 & 1\\-1 & 0\end{array}\right]$. 
 
Harmonic oscillators, in particular, are defined by having a quadratic Hamiltonian of the form $H = \frac{1}{2}x^{T}Rx$ with $R$ being a $\mathbb{R}^{n \times n}$ symmetric matrix, and a coupling operator of the form $L=\Lambda x$ with $\Lambda$ being a $\mathbb{C}^{\frac{n_w}{2} \times n}$ matrix (here $n$, $n_w$ and $n_y$ are positive even numbers). A special property of open harmonic oscillators is that the differential equations governing $x(t)$ are \emph{linear}. If we use an $n_y$-dimensional column vector $y\left(t\right)$ to incorporate all the quadratures of the output fields then, based on Eqs.~\eqref{eq:SLHqua} and \eqref{eq:outfields}, the dynamics of a set of open harmonic oscillators can be described by the following linear QSDEs \cite{JNP08}:
\begin{subequations}
\label{eq:qho}
\begin{align}
dx\left(t\right)&=Ax\left(t\right)dt+Bdw\left(t\right), \\
dy\left(t\right)&=Cx\left(t\right)dt+Ddw\left(t\right)
\end{align}
\end{subequations}
where $A$, $B$, $C$, $D$ are $\mathbb{R}^{n\times n}$, $\mathbb{R}^{n\times n_w}$, $\mathbb{R}^{n_y\times n}$  and $\mathbb{R}^{n_y\times n_w}$ matrices respectively defined in terms of $H$ and $L$ as follows:
\begin{subequations}
\label{ABCDintermsofHL}
\begin{align}
A &= 2\Theta_n\left(R+\Im\left(\Lambda^\dagger \Lambda\right)\right),\\
B &= 2i\Theta_n \left[\begin{array}{cc}
-\Lambda^{\dagger} & \Lambda^{T}\end{array}\right]\Gamma_{n},\\
C &= P_{n_{y}}^{T}\left[\begin{array}{cc}
T_{\frac{n_{y}}{2}} & 0_{\frac{n_{y}}{2}\times\frac{n_{w}}{2}}\\
0_{\frac{n_{y}}{2}\times\frac{n_{w}}{2}} & T_{\frac{n_{y}}{2}}
\end{array}\right]\left[\begin{array}{c}
\Lambda+\Lambda^{\sharp}\\
-i\Lambda+i\Lambda^{\sharp}
\end{array}\right],\\
D &= \left[\begin{array}{cc}I_{n_y} & 0_{n_y\times\left(n_w-n_y\right)}\end{array}\right],
\end{align}
\end{subequations}
with
\begin{align*}
T_{\frac{n_{y}}{2}} &=\left[\begin{array}{cc}
I_{\frac{n_{y}}{2}} & 0_{\frac{n_{y}}{2}\times\frac{\left(n_{w}-n_{y}\right)}{2}}\end{array}\right], \\
\Gamma_{m} &=P_{m} I_{\frac{m}{2}} \otimes  M, \\
M&=\frac{1}{2}\left[\begin{array}{cc}1 & i\\1 & -i\end{array}\right], \\
\end{align*}
and the symbol $P_{m}$ ($m$ is a positive even number) denotes an $m \times m$ permutation matrix defined so that if we consider a column vector $a=\left[\begin{array}{cccc}a_{1} & a_{2} & \cdots & a_{m}\end{array}\right]^{T}$,
then $
P_{m}a=\left[\begin{array}{cccccccc}a_{1} & a_{3} & \cdots & a_{m-1} & a_{2} & a_{4} & \cdots & a_{m}\end{array}\right]^{T}$. 

In this work we are primarily interested in \emph{engineering} the $A,B,C$ and $D$ matrices rather than deriving them from $H$ and $L$. When engineering, instead of using Eqs. (\ref{ABCDintermsofHL}) we instead typically use the so-called physical realizability conditions:
\begin{subequations}
\label{eq:qhoprc}
\begin{align}
&A\Theta_{n}+\Theta_{n} A^{T}+B\Theta_{n_{w}} B^{T}=0,\\
&BD^T=\Theta_{n} C^{T}\Theta_{n_{y}},
\end{align}
\end{subequations}
These are algebraic constraints, independent of $H$ and $L$, which the coefficient matrices $A,B,C$ and $D$ must obey for them to correspond to a physically realizable quantum system. They were originally derived by requiring the canonical commutation relations of $x(t)$ ($y(t)$) must hold for all times, a property enjoyed by open physical systems undergoing an overall unitary evolution \cite{GZ00,JNP08}. But since it has been proven: given a set of $A,B,C$ and $D$ matrices that satisfy Eqs.~(\ref{eq:qhoprc}) a corresponding $H$ and $L$ can always be found that satisfy Eqs.~(\ref{ABCDintermsofHL}) (e.g. see Theorem 3.4 in \cite{JNP08}).

\section{Quantum plants and coherent quantum observers}
\label{sec:QPACQO}

The primary goal of this work is to create a \emph{coherent quantum observer} which asymptotically tracks the observables of some arbitrary \emph{quantum plant} \cite{JNP08,MJ12,MEPUJ13}. We assume the {\em quantum plant} is some system of open harmonic oscillators with a set of $A_p$,$B_p$,$C_p$ and $D_p$ matrices which are known but we are unable to change. The linear QSDEs (see Section \ref{sec:OHOALQ}) for the plant is then:
\begin{subequations}
\label{eq:qsys}
\begin{align}
dx_p\left(t\right)&=A_px_p\left(t\right)dt+B_pdw_p\left(t\right), \\
dy_p\left(t\right)&=C_px_p\left(t\right)dt+D_pdw_p\left(t\right)
\end{align}
\end{subequations}
where $A_p$, $B_p$, $C_p$ are $\mathbb{R}^{n_x\times n_x}$, $\mathbb{R}^{n_x\times n_{w_p}}$ and $\mathbb{R}^{n_{y_p}\times n_x}$ matrices respectively (here $n_x$, $n_{w_p}$ and $n_{y_p}$ are positive even numbers), and $D_p=\left[\begin{array}{cc}I_{n_{y_p}} & 0_{n_{y_p}\times\left(n_{w_p}-n_{y_p}\right)}\end{array}\right]$.
Furthermore, $A_p$, $B_p$, $C_p$ and $D_p$ satisfy the following physical realizability conditions
\begin{subequations}
\label{eq:qsyscommu}
\begin{align}
&A_p\Theta_{n_x}+\Theta_{n_x} A_p^{T}+B_p\Theta_{n_{w_p}} B_p^{T}=0,\\
&B_pD_p^T=\Theta_{n_x} C_p^{T}\Theta_{n_{y_p}}.
\end{align}
\end{subequations}
 As shown in Fig. \ref{Plant_CMTObs}, we take the quantum output signal of the plant and directly fed it into the coherent quantum observer \cite{MJ12,MEPUJ13,VA14}.

An {\em (MT) coherent quantum observer} is another system of quantum harmonic oscillators which we engineer such that the system variables track those of the quantum plant asymptotically in the sense of mean values. The coherent quantum observer is driven by the output of the quantum plant directly; No measurement is involved. A coherent quantum observer has equations of the form
\begin{subequations}
\label{eq:qobs}
\begin{align}
dx_o\left(t\right)&=\left(A_p-KC_p\right)x_o\left(t\right)dt+Kdy_p\left(t\right)+B_odw_o\left(t\right),\\
dy_o\left(t\right)&=C_ox_o\left(t\right)dt+D_o\left[\begin{array}{cc}dy_p\left(t\right)^{T} & dw_o\left(t\right)^{T}\end{array}\right]^{T}
\end{align}
\end{subequations}
where: the $n_x$-dimensional column vector $x_o\left(t\right)$ denotes the \lq\lq{estimate}\rq\rq \  of $x_p\left(t\right)$; $K$, $B_o$ are $\mathbb{R}^{n_x\times n_{y_p}}$, $\mathbb{R}^{n_x\times n_{w_o}}$ matrices respectively; and $D_o$ is given by $D_o=\left[\begin{array}{cc} I_{n_{y_o}} & 0_{n_{y_o}\times\left(n_{y_p}+n_{w_o}-n_{y_o}\right)}\end{array}\right]$. Note that the system described by Eqs.~(\ref{eq:qobs}) must also satisfy the following physical realizability conditions
\begin{subequations}
\label{eq:MHOPR}
\begin{align}
&\left(A_p-KC_p\right)\Theta_{n_{x}}+\Theta_{n_{x}}\left(A_p-KC_p\right)^{T}\nonumber\\
&+K\Theta_{n_{y}}K^{T}+B_o\Theta_{n_{w_o}}B_o^{T}=0,\\
&\left[\begin{array}{cc} K & B_o\end{array}\right]D_{o}^{T}=\Theta_{n_{x}}C_{o}^{T}\Theta_{n_{y_o}}
\end{align}
\end{subequations}
which put restrictions on $K$ and $B_o$ \cite{MJ12}. In the case of $B_o\neq0$, the algebraic constraints Eqs.~\eqref{eq:MHOPR} indicate that an additional quantum noise signal $w_0(t)$ is needed. 

We use $\mu_{p}\left(t\right)$ and $\mu_{o}\left(t\right)$ to denote the first moments of the plant and the observer respectively, i.e.,
\begin{align*}
\mu_{p}\left(t\right)=&\left\langle x_p\left(t\right)\right\rangle,\\
\mu_{o}\left(t\right)=&\left\langle x_o\left(t\right)\right\rangle.
\end{align*}
The equations of motion for the first moments of the plant and the observer are:
\begin{subequations}
\label{eq:firstmom}
\begin{align}
\dot{\mu}_{p}\left(t\right)&=A_p\mu_{p}\left(t\right),\\
\dot{\mu}_{o}\left(t\right)&=\left(A_p-KC_p\right)\mu_{o}\left(t\right)+KC_p\mu_{p}\left(t\right).
\end{align}
\end{subequations}
Now we define $e_{\mu}\left(t\right)=\mu_{p}\left(t\right)-\mu_{o}\left(t\right)$ as the error which gives the difference between the first moments of the plant and the corresponding observer. And according to Eqs. (\ref{eq:firstmom}) , it evolves as
\begin{align}
\dot{e}_{\mu}\left(t\right)=\left(A_p-KC_p\right)e_{\mu}\left(t\right).
\end{align}
$e_{\mu}$ converges to zero asymptotically if and only if $A_p-KC_p$ is Hurwitz \cite{MJ12}. Hurwitz here means that all the eigenvalues of $A_p-KC_p$ have strictly negative real parts, and hence $\underset{t\rightarrow\infty}{\lim} e_{\mu}(t) = 0$.

Thus given a quantum plant described by Eqs.~(\ref{eq:qsys}), the coefficient matrices of a MT coherent quantum observer described by Eqs.~(\ref{eq:qsys}) are designed such that 
\begin{enumerate}
\item $\left(A_p-KC_p\right)$ is Hurwitz;
\item The system described by Eqs.~(\ref{eq:qobs}) corresponds to an open quantum harmonic oscillator.
\end{enumerate}

Furthermore, a MT coherent quantum observer can always be found with arbitrary rates of error convergence (proportional to the real parts of eigenvalues of $A_p-KC_p$) for a \emph{detectable} plant \cite{MJ12}. The term detectable comes from classical control \cite{KS72,GE02}, and it means that all modes of the plant are either observable or stable. Where observability means only given the outputs the state of a mode can be determined in finite time. Whether a plant is detectable or not can be judged entirely from the $A_p$ and $C_p$ matrices (e.g., see \cite{KS72,GE02} for details). Detectability then forms a sufficient (but not necessary) condition for a MT coherent observer to exist.

A MT coherent observer is limited in that it only tracks the mean values of the plant. The covariance matrix is not guaranteed to match between the plant and the MT coherent observer. This means that important properties, e.g. the entanglement or energy of the plant, may not be correctly estimated. We aim to remove these limitations and create a CMT coherent observer whose quantum state matches the plant completely in the asymptotic limit. 

\section{CMT coherent observers for open harmonic oscillators}
\label{sec:CCOFOHO}
As an extension of a MT observer, here we create a CMT observer which also tracks the covariance matrix of a quantum plant. We require that a CMT observer also tracks the mean values, thus every CMT observer is also a MT observer (but not vice versa).

Let us use $\Sigma_{p}\left(t\right)$ and $\Sigma_{o}\left(t\right)$ to denote the covariance matrix of the plant and the observer respectively, and $\Sigma_{po}$ denotes the cross variance:
\begin{align*}
\Sigma_{p}\left(t\right)=&\frac{1}{2}\left\langle x_p\left(t\right)x_p^{T}\left(t\right)+\left(x_p\left(t\right)x_p^{T}\left(t\right)\right)^{T}\right\rangle \\
&- \langle x_p\left(t\right)\rangle\langle x_p^T\left(t\right)\rangle,\\
\Sigma_{o}\left(t\right)=&\frac{1}{2}\left\langle x_o\left(t\right)x_o^{T}\left(t\right)+\left(x_o\left(t\right)x_o^{T}\left(t\right)\right)^{T}\right\rangle\\
&- \langle x_o\left(t\right)\rangle\langle x_o^T\left(t\right)\rangle,\\
\Sigma_{po}\left(t\right)=&\left\langle x_{p}\left(t\right)x_{o}^{T}\left(t\right)\right\rangle - \langle x_p\left(t\right)\rangle\langle x_o^T\left(t\right)\rangle.
\end{align*}
The evolutions for the correlation matrices of the plant and observer are given by:
\begin{subequations}
\label{eq:sigma_p}
\begin{align}
\dot{\Sigma}_{p}\left(t\right)=&A_p\Sigma_{p}\left(t\right)+\Sigma_{p}\left(t\right)A_p^{T}+B_pB_p^{T},\\
\dot{\Sigma}_{po}\left(t\right)=&A_p\Sigma_{po}\left(t\right)+\Sigma_{po}\left(t\right)\left(A_p-KC_p\right)^{T}\nonumber\\
&+\Sigma_{p}\left(t\right)\left(KC_p\right)^{T}+B_pK^{T},\\
\dot{\Sigma}_{po}^{T}\left(t\right)=&\Sigma_{po}^{T}\left(t\right)A_p^{T}+\left(A_p-KC_p\right)\Sigma_{po}^{T}\left(t\right)\nonumber\\
&+KC_p\Sigma_{p}\left(t\right)+KB_p^{T},\\
\dot{\Sigma}_{o}\left(t\right)=&\left(A_p-KC_p\right)\Sigma_{o}\left(t\right)+\Sigma_{o}\left(t\right)\left(A_p-KC_p\right)^{T}\nonumber\\
&+KC_p\Sigma_{po}\left(t\right)+\Sigma_{po}^{T}\left(t\right)\left(KC_p\right)^{T}\nonumber\\
&+KK^{T}+B_oB_o^T,
\end{align}
\end{subequations}
where $\Sigma_p\left(t\right)$, $\Sigma_o\left(t\right)$ and  $\Sigma_{po}\left(t\right)$ are real matrices with $\Sigma_p\left(t\right)$ and $\Sigma_o\left(t\right)$ nonnegative. The difference between $\Sigma_{p}\left(t\right)$ and $\Sigma_{o}\left(t\right)$ is $e_{\Sigma}\left(t\right)=\Sigma_{p}\left(t\right)-\Sigma_{o}\left(t\right)$. Then a CMT coherent quantum observer is defined as
\begin{definition} Given a system described by Eqs.~(\ref{eq:qsys}), a system described by Eqs.~(\ref{eq:qobs}) is a {\em CMT coherent quantum observer} for the system described by Eqs.~(\ref{eq:qsys}) if
\begin{enumerate}
\item The system described by Eqs.~(\ref{eq:qobs}) is a MT coherent quantum observer for the system described by Eqs.~(\ref{eq:qsys});
\item The covariance matrix of the observer described by Eqs.~(\ref{eq:qobs}) tracks that of the plant  described by Eqs.~(\ref{eq:qsys}) asymptotically, i.e.,
\begin{align*}
\underset{t\rightarrow\infty}{\lim}\Sigma_{p}\left(t\right)-\Sigma_{o}\left(t\right)=\underset{t\rightarrow\infty}{\lim}e_{\Sigma}\left(t\right)=0.
\end{align*}
\end{enumerate}
\end{definition}
Our main theorem which concerns the existence of a CMT coherent quantum observer of the form (\ref{eq:qobs}) is presented below.
\begin{theorem}
There exists a CMT coherent quantum observer described by Eqs.~(\ref{eq:qobs}) for a quantum plant described by Eqs.~(\ref{eq:qsys}) if and only if
\begin{enumerate}
\item $A_p-KC_p$ is Hurwitz;
\item The following identity
\begin{align}
\underset{s\rightarrow0}{\lim}&\left(E_{o}\otimes E_{o}-E_{p}\otimes E_{p}\right)\nonumber\\
&\times \left(sI_{4n_{x}^{2}}-I_{n_{2x}}\otimes A-A\otimes I_{n_{2x}}\right)^{-1}\nonumber\\
&\times \mathrm{vec}\left(BB^{T}\right)=0
\label{eq:secmom}
\end{align}
holds. Here
\begin{align*}
E_{p}&=\left[\begin{array}{cc}
I_{n_{x}} & 0_{n_{x}}\end{array}\right],\\
E_{o}&=\left[\begin{array}{cc}
0_{n_{x}} & I_{n_{x}}\end{array}\right],
\end{align*}
and the coefficient matrices of a joint plant-observer system are given by
\begin{align*}
A&=\left[\begin{array}{cc}A_p & 0\\KC_p & A_p-KC_p\end{array}\right],\\
B&=\left[\begin{array}{cc}B_p & 0\\KD_p & B_o\end{array}\right];
\end{align*}
\item The system described by Eqs.~\eqref{eq:qobs} is physically realizable.
\end{enumerate}
\end{theorem}
\begin{proof}
First of all, in order to ensure the convergence of $e_\mu\left(t\right)$, $A_p-KC_p$ must be Hurwitz.

The covariance matrix for the joint plant-observer system denoted $\Sigma\left(t\right)=\left[\begin{array}{cc}
\Sigma_{p}\left(t\right) & \Sigma_{po}\left(t\right)\\
\Sigma_{po}^{T}\left(t\right) & \Sigma_{o}\left(t\right)
\end{array}\right]$
 satisfies the following Lyapunov differential equation
\begin{align}
\dot{\Sigma}\left(t\right)=A\Sigma\left(t\right)+\Sigma\left(t\right)A^T+BB^T.
\label{eq:CMcl}
\end{align}
Note that
\begin{align*}
\Sigma_p\left(t\right) &= E_p\Sigma \left(t\right)E_p^T,\\
\Sigma_o\left(t\right) &= E_o\Sigma \left(t\right)E_o^T
\end{align*}
and thus
\begin{align}
\mathrm{vec}\left(\Sigma_{p}\left(t\right)\right)=E_{p}\otimes E_{p}\mathrm{vec}\left(\Sigma\left(t\right)\right),\\
\mathrm{vec}\left(\Sigma_o\left(t\right)\right)=E_o\otimes E_o\mathrm{vec}\left(\Sigma\left(t\right)\right).
\end{align}
By using the Laplace transform $\mathcal{L}\left(\cdot\right)$ to Eq.~(\ref{eq:CMcl}), we can obtain
\begin{align*}
\mathcal{L}\left(\mathrm{vec}\left(\Sigma\left(t\right)\right)\right)=&\left(sI_{4n_{x}^{2}}-I_{n_{2x}}\otimes A-A\otimes I_{n_{2x}}\right)^{-1}\\
&\times \left(\frac{\mathrm{vec}\left(BB^{T}\right)}{s}+\mathrm{vec}\left(\Sigma\left(0\right)\right)\right),
\end{align*}
then
\begin{align*}
&\mathcal{L}\left(\mathrm{vec}\left(\Sigma_{o}\left(t\right)-\Sigma_{p}\left(t\right)\right)\right)=\\
&\left(E_{o}\otimes E_{o}-E_{p}\otimes E_{p}\right)\left(sI_{4n_{x}^{2}}-I_{n_{2x}}\otimes A-A\otimes I_{n_{2x}}\right)^{-1}\\
&\times \left(\frac{\mathrm{vec}\left(BB^{T}\right)}{s}+\mathrm{vec}\left(\Sigma\left(0\right)\right)\right).
\end{align*}
Since a CMT observer has the property that $\underset{t\rightarrow\infty}{\lim}e_{\Sigma}\left(t\right)=0$, we have to require that all the poles of $\mathcal{L}\left(\rm{vec}\left(e_{\Sigma}\left(t\right)\right)\right)$ are located on the left side of the $s$-plane. Or equivalently,
\begin{align}
\underset{s\rightarrow0}{\lim}s\mathcal{L}\left(\mathrm{vec}\left(e_{\Sigma}\left(t\right)\right)\right)=&\underset{s\rightarrow0}{\lim}\left(E_{o}\otimes E_{o}-E_{p}\otimes E_{p}\right)\nonumber\\
&\times \left(sI_{4n_{x}^{2}}-I_{n_{2x}}\otimes A-A\otimes I_{n_{2x}}\right)^{-1}\nonumber\\
&\times \mathrm{vec}\left(BB^{T}\right)\nonumber\\
=&0
\end{align}
which gives Eq. (\ref{eq:secmom}).

Finally, Eqs.~(\ref{eq:qobs}) must correspond to an open harmonic oscillator, which requires that the physical realizability condition given by Eqs.~(\ref{eq:MHOPR}) should hold \cite{MJ12,JNP08}.
\end{proof}

We have found necessary and sufficient conditions for the existence of a CMT observer. But it is still a challenging task to construct a CMT observer by solving Eqs.~\eqref{eq:secmom} and~\eqref{eq:MHOPR}. Thus we consider a special case where it is easier to construct a CMT observer. Specifically, we assume $A_p$ is Hurwitz. Any plant with a unique steady state has an $A_p$ matrix which is Hurwitz.

The primary advantage of $A_p$ being Hurwitz, is that we can guarantee the existence of steady states values for all the covariance matrices (i.e., $\underset{t\rightarrow\infty}{\lim} \dot\Sigma_p(t) = 0$). Solving Eqs.~\eqref{eq:sigma_p} in steady state gives:
\begin{align}
\label{eq:CMD}
&\left(A_p-KC_p\right)e_{\Sigma}+e_{\Sigma}\left(A_p-KC_p\right)^{T}\nonumber\\
&+KC_p\left(\Sigma_p-\Sigma_{po}\right)+\left(\Sigma_p-\Sigma_{po}^T\right)\left(KC_p\right)^{T}\nonumber\\
&+B_pB_p^{T}-KK^{T}-B_oB_o^{T}=0
\end{align}
where the steady state $e_\Sigma \equiv \underset{t\rightarrow\infty}{\lim} e_{\Sigma}\left(t\right)$.

Furthermore when $A_p-KC_p$ is Hurwitz (as required for a CMT observer),  then ${e}_{\Sigma}=0$. Substituting ${e}_{\Sigma}=0$ to Eq.~\eqref{eq:CMD}  gives
\begin{align}
&KC_p\left(\Sigma_{p}-\Sigma_{po}\right)+\left(\Sigma_{p}-\Sigma_{po}^T\right)\left(KC_p\right)^{T}\nonumber\\
&+B_pB_p^{T}-KK^{T}-B_oB_o^{T}=0
\label{eq:B_o}
\end{align}
in steady state.

\begin{theorem}
Assume the quantum plant described by Eqs.~(\ref{eq:qsys}) is detectable with $A_p$ Hurwitz. The system  described by Eqs.~(\ref{eq:qobs}) is a CMT coherent quantum observer for the plant described by Eqs.~(\ref{eq:qsys}) if and only if
\begin{enumerate}
\item $A_p - KC_p$ is Hurwitz;
\item The following matrix inequality
\begin{align}
&KC_p\left(\Sigma_p-\Sigma_{po}\right) + \left(\Sigma_p-\Sigma_{po}\right)^T\left(KC_p\right)^T\nonumber\\
&+B_pB_p^T - KK^T - i K\Theta_{n_{y_p}}K^{T}\nonumber\\
&-i\left(A-KC_p\right)\Theta_{n_{_{x}}} - i\Theta_{n_{x}}\left(A-KC_p\right)^{T} \succeq0 \label{ineq:thsteady}
\end{align}

holds, where $\left(\Sigma_p-\Sigma_{po}\right)$ is the unique solution to the following Sylvester equation
\begin{align}
&A_p\left(\Sigma_{p}-\Sigma_{po}\right)+\left(\Sigma_{p}-\Sigma_{po}\right)\left(A_p-KC_p\right)^{T}\nonumber\\
&+B_pB_p^{T}-B_pK^{T} = 0.
\label{eq:Sigma_p-po}
\end{align}
\end{enumerate}
Assuming the two conditions above hold, the coupling operator characterizing the interaction between the observer and additional boson fields is then given by $L_o=\Lambda_ox_o$ where $\Lambda_o$ is any $\frac{n_{w_o}}{2}\times n_{x}$ complex matrix such that
\begin{align}
\Lambda_o^{\dagger}\Lambda_o=&-\frac{i}{4}\Theta_{n_{_{x}}}\left(A-KC_p\right)-\frac{i}{4}\left(A-KC_p\right)^{T}\Theta_{n_{x}}\nonumber\\
&+\frac{i}{4}\Theta_{n_{x}}K\Theta_{n_{y_p}}K^{T}\Theta_{n_{x}}\nonumber\\
&-\frac{1}{4}\Theta_{n_{x}}KC_p\left(\Sigma_p-\Sigma_{po}\right)\Theta_{n_{x}}\nonumber\\
&-\frac{1}{4}\Theta_{n_x}\left(\Sigma_p-\Sigma_{po}\right)^T\left(KC_p\right)^T\Theta_{n_x}\nonumber\\
&-\frac{1}{4}\Theta_{n_x}B_pB_p^T\Theta_{n_x}+\frac{1}{4}\Theta_{n_x}KK^T\Theta_{n_x}.
\end{align}
\label{thsteady}
\end{theorem}
\begin{proof}
Since the plant described by Eqs.~(\ref{eq:qsys}) is detectable, one can always find $K$ to make $A_p-KC_p$ Hurwitz. With the assumption of $A_p$ being Hurwitz, and according to Eq.~(\ref{eq:B_o}), $B_o$ must satisfy
\begin{align}
B_oB_o^T=&KC_p\left(\Sigma_{p}-\Sigma_{po}\right)+\left(\Sigma_{p}-\Sigma_{po}^T\right)\left(KC_p\right)^{T}\nonumber\\
&+B_pB_p^{T}-KK^{T},
\label{eq:solveBoCMT}
\end{align}
and the corresponding physically realizability condition is
\begin{align}
B_o\Theta_{n_{w_o}}B_o^{T}=&-\left(A_p-KC_p\right)\Theta_{n_x}-\Theta_{n_x}\left(A_p-KC_p\right)^{T}\nonumber\\
&-K\Theta_{n_{y_p}}K^{T}
\label{eq:solveBoPR}
\end{align}
Therefore $B_o$ can be determined based on Eqs. (\ref{eq:solveBoCMT}) and (\ref{eq:solveBoPR}).

In accordance with the physical form of an open harmonic oscillator described by Eqs.~(\ref{eq:qobs}) with $L_o=\Lambda_ox_o$, $B_o$ is given by \cite{JNP08,MJ12,MEPUJ13})
\begin{align}
B_o=2i\Theta_{n_{x}}\left[\begin{array}{cc}
-\Lambda_o^{\dagger} & \Lambda_o^{T}\end{array}\right]\Gamma_{n_{w_o}}.
\label{eq:Boform}
\end{align}
Here $\Gamma_{n_{w_o}}$ is defined in Section \ref{sec:OHOALQ}.

By using the form of $B_o$ given in Eq.~(\ref{eq:Boform}), we can obtain that
\begin{align}
B_oB_o^T = -4\Theta_{n_x}\Re\left(\Lambda_{0}^{\dagger}\Lambda_{0}\right)\Theta_{n_x},
\end{align}
then
\begin{align}
\Re\left(\Lambda_o^{\dagger}\Lambda_o\right) =&-\frac{1}{4}\Theta_{n_{x}}KC_p\left(\Sigma_p-\Sigma_{po}\right)\Theta_{n_{x}}\nonumber\\
&-\frac{1}{4}\Theta_{n_x}\left(\Sigma_p-\Sigma_{po}\right)^T\left(KC_p\right)^T\Theta_{n_x}\nonumber\\
&-\frac{1}{4}\Theta_{n_x}B_pB_p^T\Theta_{n_x}+\frac{1}{4}\Theta_{n_x}KK^T\Theta_{n_x}\end{align}
due to Eq.~(\ref{eq:solveBoCMT}).

Similarly, we have
\begin{align}
B_o\Theta_{n_{w_o}}B_o^{T} = 4i\Theta_{n_x}\Im\left(\Lambda_o^{\dagger}\Lambda_o\right)\Theta_{n_x},
\end{align}
then
\begin{align}
\Im\left(\Lambda_o^{\dagger}\Lambda_o\right) =& -\frac{i}{4}\Theta_{n_{_{x}}}\left(A-KC_p\right)-\frac{i}{4}\left(A-KC_p\right)^{T}\Theta_{n_{x}}\nonumber\\
&+\frac{i}{4}\Theta_{n_{x}}K\Theta_{n_{y_p}}K^{T}\Theta_{n_{x}}
\end{align}
based on Eq.~(\ref{eq:solveBoPR}).

Therefore, $\Lambda_o$ is any $\frac{n_{w_o}}{2}\times n_{x}$ complex matrix such that\begin{align}
\Lambda_o^{\dagger}\Lambda_o=&-\frac{i}{4}\Theta_{n_{_{x}}}\left(A-KC_p\right)-\frac{i}{4}\left(A-KC_p\right)^{T}\Theta_{n_{x}}\nonumber\\
&+\frac{i}{4}\Theta_{n_{x}}K\Theta_{n_{y_p}}K^{T}\Theta_{n_{x}}\nonumber\\
&-\frac{1}{4}\Theta_{n_{x}}KC_p\left(\Sigma_p-\Sigma_{po}\right)\Theta_{n_{x}}\nonumber\\
&-\frac{1}{4}\Theta_{n_x}\left(\Sigma_p-\Sigma_{po}\right)^T\left(KC_p\right)^T\Theta_{n_x}\nonumber\\
&-\frac{1}{4}\Theta_{n_x}B_pB_p^T\Theta_{n_x}+\frac{1}{4}\Theta_{n_x}KK^T\Theta_{n_x}\nonumber\\
&\succeq0\nonumber
\end{align}
and vice versa. Eq.~(\ref{ineq:thsteady}) can then be derived using the identity $-\Theta_{n_x} \Theta_{n_x} = I_{n_x}$.
\end{proof}

As studied in \cite{MJ12}, a MT coherent quantum observer can always be found if the plant described by Eqs.~(\ref{eq:qsys}) is detectable. However, as we intend to track the covariance matrix of a linear quantum plant using coherent observers at the same time, not all values of $K$ that make $A_p-KC_p$ Hurwitz are applicable to the design of a CMT coherent observer. Indeed, there are systems where mean value coherent observers exist but CMT observers can not be constructed. It is worth mentioning that $B_o$ can be $0$ if no additional noise is needed to ensure the physical realizability of an observer described by Eqs.~\eqref{eq:qobs} .

\section{Applications and examples}
\label{QCANS}
In this section, we present some numerical examples to illustrate the design and performance of CMT coherent quantum observers. We also compare the behavior of an MT vs. CMT observer.

\subsection{CMT observers vs. MT observers for a single-mode quantum harmonic oscillator}
In this example we consider tracking a single-mode Gaussian system. Consider an optical parametric oscillator as the linear quantum plant given by
\begin{subequations}
\label{eq:ex2sys}
\begin{align}
&dx_p=\left[\begin{array}{cc}-0.4 & 0\\0 & -0.6\end{array}\right]x_pdt-dw_p\\
&dy_p=x_pdt+dw_p
\end{align}
\end{subequations}
where $A_p=\left[\begin{array}{cc}-0.4 & 0\\0 & -0.6\end{array}\right]$, $B_p=-I_2$ and $C_p=D_p=I_2$.

If we choose $K$ to be $3I_2$, then using Eq.~\eqref{eq:solveBoPR} one can choose $B_o = \left[\begin{array}{cc}
1 & 0\\
0 & -2
\end{array}\right] $ to construct a MT coherent quantum observer.

However, in this case, according to Eq.~(\ref{eq:solveBoCMT}) we have
\begin{align}
B_oB_o^T=\left[\begin{array}{cc}
-1.6842 & 0\\
0 & -2.2857
\end{array}\right]
\end{align}
which is negative,  and therefore a CMT coherent observer cannot be designed with $K=3I_2$.

Alternatively, one can set $K=I_2$. First, we can calculate the steady state $\Sigma_p - \Sigma_{po}=\left[\begin{array}{cc} 1.1111 & 0\\ 0 & 0.9091 \end{array}\right]$ using Eq.~(\ref{eq:Sigma_p-po}). Then by substituting $K$ and $\Sigma_p - \Sigma_{po}$ to  Eq.~(\ref{ineq:thsteady}), we find the Eq.~(\ref{ineq:thsteady}) holds. Applying the Cholesky decomposition, one can determine
\begin{align*}
\Lambda_{0}=\left[\begin{array}{cc}
0.6742 & 0.7416i\\
0 & 0.0745
\end{array}\right].
\end{align*}
It is thus that
\begin{align*}
B_o = &2i\Theta_2\left[\begin{array}{cc}
-\Lambda_{0}^{\dagger} & \Lambda_{0}^{T}\end{array}\right]\Gamma_4 \\
=&\left[\begin{array}{cccc}
-1.4832 & 0 & 0 & 0.1491\\
0 & -1.3484 & 0 & 0
\end{array}\right].
\end{align*}
 Also, we choose the initial covariance matrix for the joint plant-observer system as
\begin{align*}
\Sigma\left(0\right)=\left[\begin{array}{cc}
1.1I_{2} & 0\\
0 & 2I_{2}
\end{array}\right]
\end{align*}
which corresponds to a Gaussian separable joint state \cite{AI05}. The initial amplitudes are $\mu_p\left(0\right) =\left[\begin{array}{cc} 1 & 1\end{array}\right]^{T}$ and $\mu_o\left(0\right) =\left[\begin{array}{cc} 0 & 0\end{array}\right]^{T}$.

We can calculate $\Sigma_p\left(t\right)$ and $\Sigma_o\left(t\right)$ explicitly by using the Laplace transform, and
\begin{align*}
e_{\Sigma}\left(t\right) &= \Sigma_p\left(t\right)-\Sigma_o\left(t\right)\\
&=\left[\begin{array}{cc}
-\frac{2}{9}e^{-\frac{9}{5}t}-\frac{7}{9}e^{-\frac{14}{5}t} & 0\\
0 & \frac{2}{11}e^{-\frac{11}{5}t}-\frac{13}{11}e^{-\frac{16}{5}t}
\end{array}\right].
\end{align*}
\begin{figure}[!htp]
\centering
\includegraphics[scale=.45]{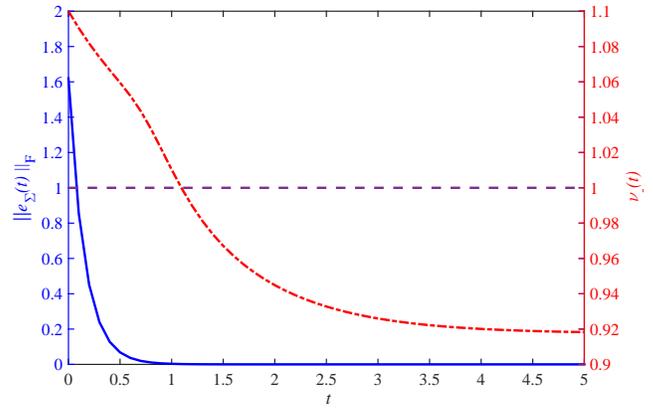}
\caption{Plot of the behavior of a CMT observer with the solid line and the dash-dot line corresponding to $||e_{\Sigma}(t)||_F$ and $\nu_-\left(t\right)$ (see Eq.~(13) in \cite{AI05}) respectively. The joint system is initialized in a Gaussian separable state.}
\label{CMT_normalex}
\end{figure}

We investigate the convergence of the covariance matrices between the plants and the coherent observers by plotting the Frobenius matrix norm of the covariance error matrix $||e_{\Sigma}(t)||_F = \sqrt{\mbox{Tr}[e_{\Sigma}^2 ]}$ against time in Fig. \ref{CMT_normalex}. When $||e_{\Sigma}(t)||_F = 0$ we can be certain $e_{\Sigma}(t)=0$, and hence the covariance matrices of the plant and observer are identical. We can see the CMT observer is performing as expected. The matrix $\Sigma_{o}\left(t\right)$ is tracking $\Sigma_{p}\left(t\right)$ asymptotically as time goes to infinity.

We also investigate the quantum correlations between the plant and the CMT observer. For Gaussian two-mode systems, entanglement is completely quantified by the smallest symplectic eigenvalue $\nu_-\left(t\right)$ of the partially transposed state, and the joint state is entangled if and only if $\nu_-\left(t\right)<1$ \cite{AI05,MJ12}. In Fig. \ref{CMT_normalex} we plot the smaller symplectic eigenvalue as a function of time. We find that the plant and the CMT observer eventually become entangled as depicted by the dash-dot line in Fig. \ref{CMT_normalex}.

As the CMT observer tracks both the first and second moments of the plant, and the quantum state is Gaussian, we expect the quantum state of the CMT observer to be identical to that of the plant. This is not guaranteed to be the case for the MT observer that only tracks the means. We compare the performance of the CMT and MT observer in this regard by plotting the quantum fidelity between the observer and the plant as a function of time in Fig. \ref{CMT_normalex_fidelity}. Quantum fidelity is widely used to quantify how close two mixed states are \cite{AI08,MM12}. For Gaussian states, the fidelity between two states can be calculated analytically (see Eq. (7) in \cite{AI08}). In this paper, we use $F\left(t\right)$ to denote the fidelity, and the closer $F\left(t\right)$ is to $1$ the more similar the two sates are to each other. In Fig. \ref{CMT_normalex_fidelity}, the state of a CMT observer (with $B_o = \left[\begin{array}{cccc} -1.4832 & 0 & 0 & 0.1491\\ 0 & -1.3484 & 0 & 0 \end{array}\right]$) gets closer to the plant state compared to a MT coherent observer (with $B_o = \left[\begin{array}{cc} 1 & 0\\ 0 & 2 \end{array}\right]$), as anticipated.
\begin{figure}[!htp]
\centering
\includegraphics[scale=.45]{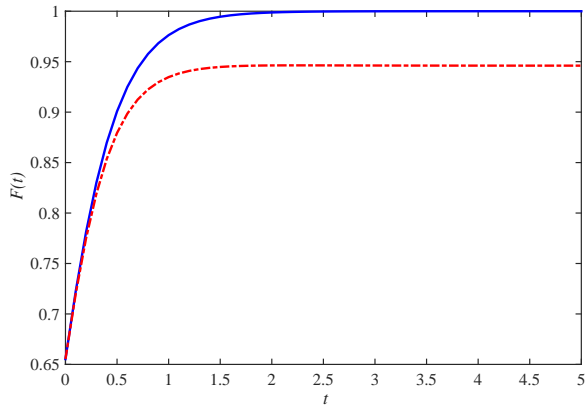}
\caption{Plot of the fidelity $F\left(t\right)$ (see Eq.~(7) in \cite{AI08}) as a function of time with the solid line and the dash-dot line corresponding to a CMT observer and a MT observer respectively.}
\label{CMT_normalex_fidelity}
\end{figure}

\subsection{Entanglement tracking of a two-mode quantum harmonic oscillator using a CMT observer}
In this example we consider a linear quantum plant which consists of two oscillators that are initially separable but eventually become entangled. The initial covariance matrix for the plant is $\Sigma_p\left(0\right)=\left[\begin{array}{cc}
1.1I_{2} & 0\\
0 & 2I_{2}
\end{array}\right]$ and its evolution is governed by the linear QSDEs
\begin{subequations}
\label{eq:sysenp}
\begin{align}
dx_{p}=&\left[\begin{array}{cccc}
-0.4 & 0 & 0 & 0\\
0 & -0.6 & 0 & 0\\
1 & 0 & -1.4 & 0\\
0 & 1 & 0 & -1.6
\end{array}\right]x_{p}dt\nonumber\\
&+\left[\begin{array}{cccc}
-1 & 0 & 0 & 0\\
0 & -1 & 0 & 0\\
1 & 0 & 1 & 0\\
0 & 1 & 0 & 2
\end{array}\right]dw_{p},\\
dy_{p}=&\left[\begin{array}{cccc}
1 & 0 & -1 & 0\\
0 & 1 & 0 & -1\\
0 & 0 & -2 & 0\\
0 & 0 & 0 & -1
\end{array}\right]x_{p}dt+dw_{p}.
\end{align}
\end{subequations}
We now design a CMT coherent observer for the plant described by Eqs.~(\ref{eq:sysenp}). One can choose the observer gain $K$ to be
\begin{align}
K=\left[\begin{array}{cccc}
0.2 & 0 & -0.1 & 0\\
0 & 0.05 & 0 & -0.1\\
0.6 & 0 & -0.1 & 0\\
0 & 0.4 & 0 & -0.1
\end{array}\right].
\end{align}
Then we find that Eq.~(\ref{ineq:thsteady}) holds, and thus a CMT observer can be constructed according to Theorem \ref{thsteady} with ($\Lambda_{o}$ is not unique)
\begin{align*}
\Lambda_{o}=\left[\begin{array}{cccc}
0.5167 & 0.5952i & -0.2914 & -0.1887i\\
0 & 0.0571 & -0.0167i & 0.1343\\
0 & 0 & 0.9316 & 0.4887i\\
0 & 0 & 0 & 0.027
\end{array}\right].
\end{align*}
We initialize the observer to $\Sigma_o\left(0\right)=2I_4$, which is different to the plant initial condition, but still separable.

We now confirm that the entanglement between the oscillators of the plant is correctly tracked by the CMT observer. In Fig. \ref{CMtracking_entanglement} we plot the smallest symplectic eigenvalue of the partially transposed state of both the observer $\nu_{-}^{o}\left(t\right)$ and the plant $\nu_{-}^{p}\left(t\right)$ as a function of time.  $\nu_{-}^{o}\left(t\right)$ converges to $\nu_{-}^{p}\left(t\right)$ asymptotically, as expected. This confirms that even quantum correlations inside the two-mode Gaussian plant can be tracked by the CMT observer. This allows for control of the plant based on quantum characteristics that were unavailable with a MT observer.  
\begin{figure}[!htp]
\centering
\includegraphics[scale=.45]{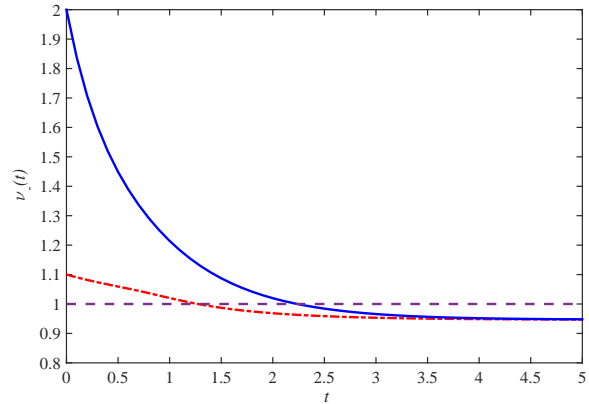}
\caption{Plot of an entanglement measure (smallest symplectic eigenvalue) of the observer $\nu_{-}^{o}\left(t\right)$ (the solid line) and the plant $\nu_{-}^{p}\left(t\right)$ (the dash-dot line) as a function of time.}
\label{CMtracking_entanglement}
\end{figure}

\subsection{Failed tracking of the covariance matrix of a singe-mode quantum harmonic oscillator}

Consider a plant with the following linear QSDEs
\begin{subequations}
\label{eq:sysneverCMT}
\begin{align}
dx_{p}&=\left[\begin{array}{cc}
-1 & 1\\
1 & -1
\end{array}\right]x_{p}dt+\left[\begin{array}{cc}
-\sqrt{2} & 0\\
0 & -\sqrt{2}
\end{array}\right]dw_{p},\\
dy_{p}&=\left[\begin{array}{cc}
\sqrt{2} & 0\\
0 & \sqrt{2}
\end{array}\right]x_{p}dt+dw_{p}
\end{align}
\end{subequations}
For the quantum plant given by Eqs.~(\ref{eq:sysneverCMT}) (Note that $A_p$ is not Hurwitz), no matter what values we choose $K$ and $B_o$ to be, Eq.~(\ref{eq:secmom}) cannot be satisfied. It is thus that a CMT coherent quantum observer can never be designed for the plant given by Eqs.~(\ref{eq:sysneverCMT}).

There do exist plants which cannot be tracked by a MT observer. For instance, certain undetectable plants cannot be tracked even in the mean values sense. However, in this example, a MT observer can be constructed even though a CMT observer cannot be. Specifically, we can choose $K = I_2$ to make  $A_p-KC_p$ Hurwitz, and then $B_o$ is determined as shown in \cite{MJ12}. Therefore, in this case, we are only able to approach the mean values of the quantum plant without tracking its covariance matrix. This demonstrates that there are additional constraints when constructing a CMT observer compared to a MT observer.

\section{Discussion and Outlook}
\label{sec:D}
We have created a CMT observer that tracks both the mean values and covariances of a system of linear quantum oscillators. We first emphasize that the CMT observer (and the previously developed MT observer) does not just have the same steady state as the system it is tracking, but it is also correlated to the plant's state. This is shown in Fig. 2 for the CMT observer, where we see it is entangled with the plant.

In future work it will be investigated whether an observer can be used in a feedback loop to better control the behavior of a plant than would otherwise be possible with direct feedback. It has already been shown that coherent feedback has advantages over measurement-based feedback. We now expect that coherent observer mediated feedback will have an advantage over direct feedback in the same sense a classical observer mediated feedback has an advantage over direct feedback.

Furthermore, we expect that the CMT observer will be more useful than a MT observer as it tracks the covariance matrix of the plant as well. Important properties of a quantum system, such as its energy, entanglement and other quantum correlations are a function of the system's covariance matrix rather than merely its mean values. 

Our work on a CMT coherent observer also raises some interesting fundamental questions with regard to engineering quantum systems in comparison to classical system. Classifying what plants can or cannot be tracked with a MT observer appears to be identical to classical observer theory. Namely a $K$ must be found such that $A_p-KC_p$ is Hurwitz. There is well established classical theory which then relates this requirement to notions such as observability and detectability \cite{KS72,AM79,GE02}. A CMT observer, on the other hand, has additional requirements which are fundamentally quantum in origin. Namely, Eq.~(\ref{ineq:thsteady}) must be satisfied in addition to $A_p-KC_p$ being Hurwitz. It remains an open question on how to interpret this additional requirement and if the classical notions of observability and detectability can be appropriately extended when discussing the tracking of a quantum plant's covariance matrix. As we are attempting to copy the entire quantum state of the plant with a CMT observer (unlike a MT observer), there may be some connection between these additional requirements and the no-cloning theorem \cite{SLGA05,KZ09}.

Outside of quantum engineering, the design and implementation of a CMT observer also looks to provide some insight into quantum measurement. When the output of the plant is measured, an optimal estimate of the quantum state of the plant can be calculated using the Belavkin-Kalman filter (also referred to as stochastic trajectories) \cite{WM10}. However, research suggests the situation becomes much more complicated when there is no measurement step. It has been proven the Belavkin-Kalman filter fails in the presence of a fully quantum non-commutative output signal \cite{Amini:2014,VPB93,BHJ07} and furthermore measurement-based Kalman filters are challenging to be realized efficiently with quantum hardware \cite{Hush:2013a}. The CMT observer is the first coherent method of providing an estimate of the full quantum state of a plant. Note we never invoked the measurement postulate when deriving the CMT observer. It is entirely derived in the framework of open systems. Creating and better understanding estimators for quantum systems which do not explicitly require the measurement postulate is an important part of further refining our understanding of quantum measurement. 

\begin{acknowledgments} 

ZM and MJ acknowledges support by the Australian Research Council Centre of Excellence for Quantum Computation and Communication Technology (project number CE110001027) and Air Force Office of Scientific Research (project number FA2386-09-1-4089). MH acknowledges funding from Australian Research Council Discovery Project (project number DP110102322).

\end{acknowledgments} 

\appendix

\section{Notation}
\label{sec:app}
In this paper the asterisk is used to indicate the Hilbert space adjoint $X^{\star}$ of an operator $X$, as well as the complex conjugate $z^{\star}=x-iy$ of a complex number $z=x+iy$ (here, $i=\sqrt{-1}$ and $x,y$ are real). Real and imaginary parts are denoted by $\Re\left(z\right)=\frac{z+z^{\star}}{2}$ and $\Im\left(z\right)=\frac{z-z^{\star}}{2i}$ respectively. The conjugate transpose $A^\dagger$ of a matrix $A=\left\{ a_{ij}\right\} $ is defined by $A^{\dagger}=\left\{ a_{ji}^{\star}\right\} $. The conjugate  $A^{\sharp}=\left\{ a_{ij}^{\star}\right\} $ and the transpose $A^{T}=\left\{ a_{ji}\right\} $ of a matrix is defined so that $A^{\dagger}=\left(A^T\right)^{\sharp}=\left(A^{\sharp}\right)^T$. $\textrm{det}\left(A\right)$ denotes the determinant of a matrix $A$, and $\textrm{tr}\left(A\right)$ represents the trace of $A$.  vec$\left(A\right)$ denotes the vectorization of a matrix $A$. $\left\Vert A\right\Vert _{F}$ denotes the Frobenius norm, i.e., $\left\Vert A\right\Vert _{F}=\sqrt{\mathrm{tr}\left(A^{\dagger}A\right)}$. The mean value (quantum expectation) of an operator $X$ in the state $\rho$ is denoted by $\left\langle X \right\rangle = \mathbb{E}_{\rho}\left[X\right]=\textrm{tr}\left(\rho X\right)$. The commutator of two operators $X,Y$ is defined by $\left[X,Y\right]=XY-YX$. The anticommutator of two operators $X,Y$ is defined by $\left\{X,Y\right\}=XY+YX$. The tensor product of operators $X,Y$ defined on Hilbert spaces $\mathbb{H},\mathbb{G}$ is denoted $X \otimes Y$, and is defined on the tensor product Hilbert space $\mathbb{H}\otimes\mathbb{G}$. $I_n$ ($n\in\mathbb{N}$) denotes the $n$ dimensional identity matrix. $0_n$ ($n\in\mathbb{N}$) denotes the $n$ dimensional zero matrix.

\bibliography{CMTreferences}

\end{document}